\documentclass{rjparticle}
\usepackage{graphicx}

\newtheorem{theorem}{Theorem}

\newenvironment{proof}[1][Proof]{\begin{trivlist}
\item[\hskip \labelsep {\bfseries #1}]}{\end{trivlist}}

\newenvironment{example}[1][Example]{\begin{trivlist}
\item[\hskip \labelsep {\bfseries #1}]}{\end{trivlist}}
\newenvironment{remark}[1][Remark]{\begin{trivlist}
\item[\hskip \labelsep {\bfseries #1}]}{\end{trivlist}}
\newenvironment{corollary}[1][Corollary]{\begin{trivlist}
\item[\hskip \labelsep {\bfseries #1}]}{\end{trivlist}}

\newcommand{\qed}{\nobreak \ifvmode \relax \else
      \ifdim\lastskip<1.5em \hskip-\lastskip
      \hskip1.5em plus0em minus0.5em \fi \nobreak
      \vrule height0.75em width0.5em depth0.25em\fi}

\newcommand{\miktex}{\hbox{Mik\kern-.15em\TeX}}

\title{Recurrence relations for the number of solutions \\ 
of a class of Diophantine equations}
\author[1,2]{M. I. Krivoruchenko}
\affil[1]{
Institute for Theoretical and Experimental Physics$\mathrm{,}$ B. Cheremushkinskaya 25 \\ 
117218 Moscow, Russia }
\affil[2]{Department of Nano-$\mathrm{,}$ Bio-$\mathrm{,}$ Information and Cognitive Technologies\\ 
Moscow Institute of Physics and Technology$\mathrm{,}$ 9 Institutskii per. \\ 
141700 Dolgoprudny$\mathrm{,}$ Russia}

\begin{document}
\maketitle
\begin{abstract}
Recursive formulas are derived for the number of solutions of linear and 
quadratic Diophantine equations with positive coefficients. This result is further 
extended to general non-linear additive Diophantine equations. It is shown that all three 
types of the recursion admit an explicit solution in the form of complete Bell polynomial, 
depending on the coefficients of the power series expansion 
of the generating functions for the sequences of individual terms in the Diophantine equations.
\end{abstract}

\section{Introduction}

Diophantine equations are encountered in theory of partitions, combinatorial analysis, 
integer linear programming, and in many related areas \cite{MORD1969, NATH2000}.
Although not as common, Diophantine equations occur in physical applications. 
We mention the problems in solid state physics \cite{AVRO1986}-\cite{PELE1998}
and theory of angular momentum \cite{BREM1986,SRIN1988}.
A new application area emerged recently in the field of nuclear physics in connection with the problem 
of calculating degeneracy of the symmetry group reduction chains \cite{Gheo2004}.

In this paper the problem of estimating the number of solutions of Diophantine equations 
is discussed from the viewpoint of Ward identities,
successfully used earlier to establish relationships between different 
kinds of invariant integrals on continuous groups (see, e.g., \cite{WHLM2007}).
In Sects.~2 and~3, we consider linear and quadratic Diophantine equations and 
derive recursive formulas for the number of solutions. 
In Sect. 4 we show that the number of solutions of general non-linear additive Diophantine 
equations can also be calculated recursively, 
which involves the partial Bell polynomials 
evaluated at 
the first expansion coefficients of the generating functions
for the sequences of individual terms in the equation.
In conclusion of Sect. 4 we show, furthermore, 
that the number of solutions is equal to
the complete Bell polynomial 
evaluated at
the first expansion coefficients of the logarithm of the full generating function.

\section{Linear Diophantine equation}

We consider the linear Diophantine equation
\begin{equation}
a_{1}k_{1}+\ldots +a_{r}k_{r}=n  \label{1}
\end{equation}
with positive integer coefficients $a_{l} \in \mathbb{N}$ ($l=1,\ldots ,r$) and an integer $n$. 
We are looking for the number $\nu_{r}(n)$ 
of non-negative integer solutions $k_{l} \in \mathbb{N}_0$ ($l=1,\ldots ,r$) 
of this equation.

\begin{theorem}
The number of non-negative solutions of the linear Diophantine
equation (\ref{1}) with positive integer coefficients $a_{l}$ can be
obtained by the following recursive formula: 
\begin{equation}
\nu _{r}(n)=\frac{1}{n}\sum_{l=1}^{r}a_{l}\sum_{i=1}^{[n/a_{l}]}\nu_{r}(n-ia_{l}),  
\label{RE1}
\end{equation}
with the initial conditions $\nu _{r}(n)=0$ for $n<0$ and $\nu _{r}(0)=1$.
\label{THEOREM1}
\end{theorem}

\begin{proof}
The number of solutions can be found by enumerating all non-negative values $k_l$
and selection of those combinations that satisfy Eq.~(\ref{1}).  
This is achieved through the following algebraic structure:
\begin{equation}
\nu _{r}(n)=\sum_{k_{1}\ldots k_{r}}\delta (a_{1}k_{1}+\ldots +a_{r}k_{r}-n).
\label{2}
\end{equation}
The selection of the suitable combinations is carried out with the use of the Kronecker delta 
\[
\delta (m-n)=\left\{ 
\begin{array}{ll}
1, & m=n, \\ 
0 & m\neq n.
\end{array}
\right. 
\]

We represent the Kronecker delta in the form of a contour intergal: 
\[
\delta (m-n)=\frac{1}{2\pi i}\oint \frac{dz}{z^{n+1}}z^{m}, 
\]
where the integration is counterclockwise in a neighborhood of $z=0$. 
With the use of this representation, Eq.~(\ref{2}) can be written
in the form 
\begin{equation}
\nu _{r}(n)=\frac{1}{2\pi i}\oint \frac{dz}{z^{n+1}}\sum_{k_{1}\ldots
k_{r}}z^{a_{1}k_{1}+\ldots +a_{r}k_{r}}.  \label{2a}
\end{equation}
For $|z| < 1$ the series converges, and the order of summation and integration can be changed.
The summations over $k_{l}$ are independent. There are $r$ of the
summations, and each is a geometric progression. The integral takes the form
\begin{equation}
\nu _{r}(n)=\frac{1}{2\pi i}\oint \frac{dz}{z^{n+1}}\varphi (z),  
\label{3}
\end{equation}
where
\begin{equation}
\varphi (z)=\prod_{l=1}^{r}\frac{1}{1-z^{a_{l}}}.  
\label{GENE}
\end{equation}
The function $\varphi (z)$ is analogous to Euler's generating function in the number theory and 
the probability generating function in the theory of probability.
Expression (\ref{3}) gives the derivative of $\varphi (z)$: 
\[
\nu _{r}(n) = \frac{1}{n!} \frac{d^{n}}{dz^{n}}\varphi (0). 
\]
Inside the unit circle the integrand $\varphi (z)/z^{n + 1}$ is the analytic function 
with the pole at $z=0$. To get a recursion, one can proceed by any of the following ways:

Firstly, one can integrate by parts. Secondly, 
one can exploit the analyticity of $\varphi (z)/z^{n + 1}$. 
The radius of the path of integration in the neighborhood $z=0$ does not affect the result.
Introducing the radius of the circle explicitly and differentiating it, 
we obtain a recursive formula that coincides with the formula obtained by using the integration by parts.
Such techniques were used earlier, e.g. in Refs. \cite{RADU2012,KRIV2012}, to get recursive formulae 
for the normalization of particle-number-projected BCS wave function and for the probability
distribution of the number of electron-positron pairs created in an external
electric field.
Thirdly, the integration over $dz/z$ is in fact the invariant integration in $U(1)$ group. By properties of the group integral, phase
transformation $z\rightarrow ze^{i\chi }$ does not affect the value of integral. This yields an identity 
\begin{equation}
\nu _{r}(n) = \frac{1}{2\pi i}\oint \frac{dz}{nz^{n}} \frac{d \ln(\varphi (z))}{dz} \varphi (z), 
\label{recu}
\end{equation}
which is the special case of Ward identity. 

The terms $1/(1-z^{a_{k}})$ originating from the logarithmic derivative of $\varphi (z)$
can be expanded in a power series over $z^{a_{k}}.$ The series is truncated bacause the singularity 
of the integrand is a finite-order pole at $z=0$. 
The expansion terms of very high order remove such a singularity, so that the corresponding contour integrals vanish.
The sum over $i$
is therefore within the limits $1\leq i\leq [n/a_{l}].$ Comparing the
different terms with Eq.~(\ref{3}), we notice that each term of the
expansion represents $\nu _{r}(m)$ for some $m<n$. We thus arrive at Eq.~(\ref{RE1}).

The initial condition for the recursion $\nu _{r}(0)=1$ is obvious. There
are no solutions of the equation for negative $n$, so $\nu _{r}(n)=0$ for $n<0$. 
\end{proof}

Theorem 1 has
 
\begin{corollary}
The recursive formula (\ref{RE1}) can be written in the form 
\begin{equation}
\nu _{r}(n)=\frac{1}{n}\sum_{m = 1}^{n}\rho (m)\nu _{r}(n-m),  \label{21}
\end{equation}
where 
\begin{equation}
\rho (m)=\sum_{l=1}^{r}a_{l}\sum_{i \geq 1}\delta (m-ia_l)  \label{22}
\end{equation}
is the sum of coefficients $a_{l}$ that are divisors of $m$.
\end{corollary}


\begin{remark}
Theorem 1 generalizes Theorem 15.1 of Ref. \cite{NATH2000}. 
\end{remark}

\begin{remark}
A partition of $n$ is a nonincreasing sequence of positive integers whose
sum equals $n$. For $a_l = l$, the number of partitions of $n$ is given by
the corresponding coefficient in the expansion of the generating function (%
\ref{GENE}) in a power series in the neighborhood of $z=0$ (Euler's theorem). 
This number is equal to the number of distinct non-negative solutions of the Diophantine
equation (\ref{RE1}) with the coefficients $a_l = l$. The one-to-one
correspondence between the partitions of $n$ and the solutions of 
the Diophantine equation (\ref{RE1}) is achieved
by interpreting the variable $k_l$ as the number of times of occurrence 
of the number $l$ in the partition of $n$.
\end{remark}

\begin{remark}
In the asymptotic regime $n\to \infty $, the summation over the index $i$ in
Eq.~(\ref{RE1}) can be replaced by an integral. The derivative of Eq.~(\ref{RE1})
in $n$ leads then to an equation that has the solution 
\[
\nu _{r}(n)\sim C_{r}n^{r-1}. 
\]
The coefficient $C_{r}$ can be found from
\[
\nu _{r}(n)=\sum_{i=0}^{[n/a_{r}]}\nu _{r-1}(n-ia_{r}). 
\]
In the continuum limit, the sum is replaced by an integral, which gives 
$C_{r}^{-1}=(r-1){a_{r}}C_{r-1}^{-1}.$ By lowering further the index $r$, 
we obtain, for arbitrary coefficients $a_{l}$, 
\[
C_{r}^{-1}=(r-1)!\prod_{l=1}^{r}{a_{l}}C_{0}^{-1}. 
\]
In the case where $a_{l}$ are coprime, $C_{0}=1$. \cite{SCHU26}

If the coefficients contain only one common divisor, $d$, and $a_l/d$ are coprime,
the problem reduces to the case of the coprime coefficients. In such a case, 
$C_{0} = d$, when the $n = 0$~(mod~$d$) and $C_{0} = 0$, when $n \neq 0$~(mod~$d$). 
After averaging $n$ over an interval $\Delta n > d$, we get $C_{0} = 1$. 
The equality $C_{0} = 1$ also holds when only a part of the coefficients 
has a single common divisor. To see this, we use the equation 
\[
\nu _{p+q}(n) = \sum_{s=0}^{[n/d]} \nu_p(n - ds)\nu_q(s). 
\]
A similar recursion can be found in Ref. \cite{Gheo2004}.
The first $p$ coefficients are coprime, while the last $q$ coefficients have
a single common divisor $d$. $\nu_q(s)$ counts the number of solutions of
Eq.~(\ref{1}) with the integer coprime coefficients $a_l/d$ ($l=p+1,\ldots,p+q$).
\end{remark}

\begin{example}
Consider the problem of finding the number of distinct terms in the
expansion of the determinant in the sum of products of traces of
powers of the matrix. The determinant is represented as follows \cite
{KOND1992} 
\begin{equation}
\det \left\| A\right\| =\sum_{k_{1}\ldots k_{n}}\prod_{l=1}^{n}\frac{%
(-1)^{k_{l}+1}}{l^{k_{l}}k_{l}!}\mathrm{tr}(A^{l})^{k_{l}}\mathrm{,}
\label{trace}
\end{equation}
where the admissible sets of non-negative $k_{l}$ over which we take the
summation are determined by solutions of Eq.~(\ref{1}) with $a_{l}=l$ and $%
r=n.$ The number of the various terms is given by Eq.~(\ref{RE1}). For
matrices with low dimensionality, we obtain $\nu
_{n}(n)=2,3,5,7,11,15,22,\ldots $ for $n=2,3,4,5,6,7,8,\ldots $,
respectively. As previously noted, the number of solutions of Eq.~(\ref{RE1}%
) with the coefficients $a_{l}=l$ coincides with the number of partitions of 
$n$. The asymptotic behavior of $\nu_{n}(n)$ for $n \to \infty$ has the form 
\cite{NATH2000} 
\[
\nu_{n}(n) \sim \frac{1}{4n\sqrt{3}}\exp(\pi\sqrt{2n/3}). 
\]

The number of terms on the right side of Eq.~(\ref{trace}) is
growing sub-exponentially. 
In terms of computing, the most economical method to calculate determinants 
is the Gauss elimination method, 
which requires a polynomially large number of operations. The decomposition (\ref{trace}) is of interest 
when there is a symmetry and a need to preserve it at each step of the calculation.
\end{example}

\begin{example}
A similar combinatorial problem arises in calculating the derivative of
composite functions. This problem leads to the Fa\`a di Bruno's formula (see, e.g., \cite{GRAD2007}) 
\begin{equation}
\frac{1}{r!}\frac{d^{r}}{dz^{r}}f(g(z))=\sum_{k_{1}...k_{r}} f^{(k_{1} +
k_{2} ... + k_{r})}(g(z)) \prod_{l=1}^{r} \frac{1}{{l!}^{k_{l}} k_{l}!}
(g^{(l)}(z))^{k_{l}} ,  
\label{deriva}
\end{equation}
where $f^{(m)}$ and $g^{(m)}$ are the $m$-order derivatives; the summation
is over all sets of the non-negative $k_{1},...,k_{r}$ 
that satisfy Eq.~(\ref{1}) with $a_l = l$. 
The number of terms in the right side of Eq.~(\ref{deriva}) and the asymptotic behavior
are those of the trace decomposition (\ref{trace}).
\end{example}

\begin{example}
As another application, one can mention the Bell polynomials \cite{BELL28}
in terms of which the right sides of Eqs.~(\ref{trace}) and (\ref{deriva}) can be expressed. 
The number of terms in the $n$-th complete Bell polynomial is equal to the number of solutions
to Eq.~(\ref{1}) with $a_l = l$ and $r = n$.
\end{example}

\begin{example}
For the case of $a_l = 1$ one may find an explicit expression for $\nu_{r}(n)$ from Eq.~(\ref{2a}). 
By moving the contour to infinity, we obtain 
\begin{equation}
\nu_{r}(n) = \frac{(n + r - 1)!}{n!(r - 1)!}.  
\label{simple}
\end{equation}
For $a_l = 1$ this equation solves the recursion (\ref{RE1}). 
It can be noted that $\nu_{r}(n)$ is equal to 
the number of independent components of a rank-$n$ symmetric tensor in the space of dimension $r$.
The solutions $a_l = 1$
differ from the solutions $a_l = l$ combinatorially in the sense that all
the partitions of $a_l = 1$ are considered to be different. For instance, 
$1 + 1 + 3 = 5$ and $1 + 3 + 1 = 5$ are counted as distinct partitions of 5, 
whereas in the case of $a_l = l$ these partitions are counted as the one with 
$k_1 = 2,\;k_2 = 0,\;k_3 = 1$.
\end{example}

\begin{example}
We consider the random walk of a particle on the one-dimensional lattice. 
Suppose that the probability distribution in one step is described 
by the Poisson distribution 
\begin{equation}
p_1(n) = \frac{\alpha^n}{n!}\exp(-\alpha).  
\label{h1}
\end{equation}
The parameter $\alpha$ is the average displacement in one step. 
Condition $a_l > 0$ means that the particle is moving forward.
The probability generating function is of the form 
\begin{equation}
\varphi(z) = \exp (\sum_{l=1}^{r}(z^{a_{l}}-1)\alpha ).
\label{h2}
\end{equation}
The probability of displacement in $r$ steps by distance $n$ 
is given by the right side of Eq.~(\ref{3}). Applying the Ward identity, 
we find a recursive formula
\begin{equation}
p_{r}(n)=\frac{\alpha }{n}\sum_{l=1}^{r}a_{l}p_{r}(n-a_{l}).  
\label{h3}
\end{equation}
The initial conditions are $p_{r}(n) = 0$ for $n<0$ and $p_{r}(0) = (p_{1}(0))^r =\exp(-\alpha r)$. 
\end{example}

\section{Quadratic Diophantine equation}

We consider the quadratic Diophantine equation 
\begin{equation}
a_{1}k_{1}^{2}+\ldots +a_{r}k_{r}^{2}=n  
\label{DE2}
\end{equation}
for positive integer coefficients $a_l \in \mathbb{N}$ 
and $k_{l} \in \mathbb{Z}$ ($l=1,\ldots ,r$). 
The main result of this section can be summarized in

\begin{theorem}
The number of solutions of the Diophantine equation (\ref{DE2}) 
with positive coefficients $a_{l}$ can be obtained with the help 
of the recursive formula 
\begin{equation}
\nu _{r}(n) = \frac{1}{2n}\sum_{l=1}^{r}a_{l}\sum_{pq}
\left( -1+(-1)^{p-1}+2(-1)^{q-1}+2(-1)^{p+q}\right) p \nu_{r}(n-a_{l}pq).
\label{RE2}
\end{equation}
The summation is over the indices $l$ from one to $r$ and $p$ and $q$ from one to $pq \leq [n/a_l]$.
The initial condition is $\nu _{r}(0)=1$, and $\nu _{r}(n)$ is taken to be 0 if $n<0$.
\label{THEOREM2}
\end{theorem}

\begin{proof} 
The number of solutions of (\ref{DE2}) is written as a contour integral
\begin{equation}
\nu _{r}(n)=\frac{1}{2\pi i}\oint \frac{dz}{z^{n+1}}\sum_{k_{1}\ldots
k_{r}}z^{a_{1}k_{1}^{2}+\ldots +a_{r}k_{r}^{2}}.  
\label{DE2.cont}
\end{equation}
This representation is valid for $a_l > 0$, 
because the series with a non-positive $a_l$ diverges, when the series with positive $a_l$ 
converge, and \textit{vice versa}.

The peculiarity of the recursion is the generating function
\[
\varphi(z) = \prod_{l = 1}^{r} \vartheta (1,z^{a_l}), 
\]
expressed in terms of the product of Jacobi theta functions. Using the method of Sect. 2, we obtain 
\begin{eqnarray}
\nu _{r}(n) = \frac{1}{2\pi i}\oint \frac{dz}{nz^{n}}
\sum_{l=1}^{r} a_{l}z^{a_{l}-1} \frac{\vartheta ^{\prime }(1,z^{a_{l}})}{\vartheta (1,z^{a_{l}})}\varphi(z).
\label{DE2AA}
\end{eqnarray}
The triple product identity
\[
\vartheta (w,u)=\prod\limits_{m=1}^{\infty}(1-u^{2m})(1+w^{2}u^{2m-1})(1+w^{-2}u^{2m-1})
\]
allows to get the series expansion of the logarithmic derivative 
in the neighborhood of $u=0$: 
\begin{eqnarray*}
\frac{\vartheta ^{\prime }(1,u)}{\vartheta (1,u)} =
\sum_{m=1}^{+\infty }\sum_{s=0}^{+\infty }u^{2m(s+1)-1}\left((-2m)+2(2m-1)(-1)^{s}u^{-(s+1)}\right)
\end{eqnarray*}
Substituting this expression for $u = z^{a_l}$ in Eq.~(\ref{DE2AA}), we arrive at Eq.~(\ref{RE2}).
\end{proof}

\begin{example}
Using Eq.~(\ref{RE2}), we obtain for $a_l = 1$ $\nu_{2}(n) = 1,4,4,0,4,8,0,0,4,8$ and 
$\nu_{3}(n) = 1,6,12,8,6,24,24,0,12,30,24$ for $n = 0,1,\ldots,10$, 
which is in the agreement with the direct expansion of the generating functions.
\end{example}

\section{Diophantine equation of general additive form}

Now consider the general case:
\begin{equation}
g_{1}(k_{1})+\ldots +g_{r}(k_{r})=n,
\label{DE3}
\end{equation}
for $k_l \in \mathbb{N}_0$. The number of solutions, $\nu_r(n)$, is finite provided $g_{l}(k) \in \mathbb{N}_0$. 
We also consider the case of increasing functions: $\forall l$
$k<m$ iff $g_{l}(k)< g_{l}(m)$. Without loss of generality, one can assume $g_{l}(0) = 0$. 

\begin{theorem} Under the specified conditions the number of solutions of Eq.~(\ref{DE3})  
can be calculated from the recursion
\begin{equation}
\nu _{r}(n)=\frac{1}{n}\sum_{l=1}^{r}\sum_{m=1}^{n}\frac{1}{(m-1)!}%
K_{m}(c_{l1},\ldots ,c_{lm})\nu _{r}(n-m),
\label{RE3}
\end{equation}
where
\[
K_{n}(c_{l1},\ldots ,c_{ln})=\sum_{k=1}^{n}(-1)^{k-1}(k - 1)!B_{n,k}(1!c_{l1},\ldots,(n-k+1)!c_{ln-k+1}),
\]
$B_{n,k}(x_{1},\ldots,x_{n-k+1})$ are the partial Bell polynomials, and  
$c_{lk}$ are the expansion coefficients of the generating functions 
for the sequences of individual terms in Eq.~(\ref{DE3}):
\begin{equation}
\varphi _{l}(z)\equiv \sum_{k=0}^{\infty }z^{g_{l}(k)} = 1+\sum_{k=1}^{\infty } c_{lk} z^{k}.
\label{c}
\end{equation}
The initial conditions are as follows: $\nu _{r}(n)=0$ for $n<0$ 
and $\nu _{r}(0)=1$.
\end{theorem}

\begin{proof}
The number of solutions is calculated by the same method as in the previous two
sections. 
The expansion coefficients of the generating functions (\ref{c}) are given by:
\[
c_{lk}=\left\{ 
\begin{array}{ll}
1, & \exists m > 0:k=g_{l}(m), \\ 
0, & \forall m > 0:k\neq g_{l}(m).
\end{array}
\right. 
\]

The logarithmic derivative of the full generating function 
\begin{equation}
\varphi (z)=\prod_{l=1}^{r}\varphi _{l}(z)
\label{char}
\end{equation}
enters the integral form of $\nu_r(n)$. For the logarithm of $\varphi_{l}(z)$, we have the following
representation
\begin{equation}
\ln \left( \varphi_{l}(z)e^{C}\right) =\int_{0}^{\infty }\frac{d\xi }{\xi }%
\left( -\exp (-\xi \varphi_{l} (z))+\frac{1}{1+\xi }\right) ,  \label{inte}
\end{equation}
where $C = 0.577\ldots$ is the Euler constant. The exponent is expanded in a neighborhood of $z=0$ 
in terms of the partial Bell polynomials \cite{BELL28}
\[
\exp (-\xi \varphi_{l} (z))=\exp (-\xi )\sum_{n=0}^{\infty }\frac{1}{n!}%
\sum_{k=1}^{n}(-\xi )^{k}B_{n,k}(1!c_{l1},\ldots ,(n-k+1)!c_{ln-k+1})z^{n}
\]
Taking the derivative on both sides of Eq. (\ref{inte}), 
changing the order of summation and integration, and integrating over $\xi $, 
we obtain
\[
\frac{\varphi_{l} ^{\prime }(z)}{\varphi_{l} (z)}=\sum_{n=1}^{\infty }\frac{1}{(n-1)!%
}K_{n}(c_{l1},\ldots ,c_{ln})z^{n-1}.
\]
The integral representation (\ref{recu}) gives then the recursion (\ref{RE3}).
\end{proof}

\begin{remark}
The above scheme is limited by the requirement of positive $g_{l}(k_{l})$. 
This constraint can be eliminated provided the additivity holds. Let us consider the equation
\begin{equation}
g_{1}(k_{1})+\ldots +g_{r}(k_{r})=g_{r+1}(k_{r+1})+\ldots +g_{r+s}(k_{r+s}),
\label{DE3P}
\end{equation}
with non-negative functions $g_{l}(k_{l})$ ($l = 1,\ldots,r+s$). 
If some functions in equation (\ref{DE3}) are negative, they can be placed in the right side, 
in which case we arrive at Eq.~(\ref{DE3P}).
The number of solutions of Eq.~(\ref{DE3P}) is the same as the number of solutions in the system of two equations, 
Eq.~(\ref{DE3}) and 
\begin{equation}
g_{r+1}(k_{r+1})+\ldots +g_{r+s}(k_{r+s}) = n,
\label{DE3PP}
\end{equation}
where the right side parameter is not fixed. We have compare solutions of equations (\ref{DE3}) and (\ref{DE3PP}) 
for all values of the parameter $n$. The number of solutions of these equations 
equals $\nu_r(n)$ and $\nu_s(n)$, respectively. 
The number of solutions of the system is the product $\nu_r(n) \nu_s(n)$ summed over all $n$, and it can diverge.
\end{remark}

\begin{example}
We illustrate the method by finding few lowest solutions of the equation
\begin{equation}
k_1^3 + k_2^3 = k_3^2.
\label{332}
\end{equation}
The right side is considered as a parameter $n$. 
If the number of solutions for a square $n$ is different from zero, 
we get the proof on the existence of the solution.
Bell polynomials are programmed as standard functions with Maple 15. 
Using the recursion (\ref{RE3}), we obtain $\nu_2(1) = \nu_2(8) = \nu_2(9) = 2$,
$\nu_2(2) = \nu_2(16) =1$ and $\nu_2(n) = 0$ in other cases, for $n = 1, \ldots, 50$. 
The first two solutions of Eq.~(\ref{332}) correspond to 
$k_1 = 1$, $k_2 = 2$, $n = 3^2$ and $k_1 = 2$, $k_2 = 2$, $n = 4^2$.
\end{example}

A slight modification of the arguments leads to 

\begin{theorem}
Let $d_{lk}$ be expansion coefficients of the logarithm of the generating functions (\ref{c}). 
The expansion coefficients of the logarithm of the full generating function
(\ref{char}) are given then by
\begin{equation}
d_{k} = \sum_{l=1}^{r}d_{lk},
\label{dsum}
\end{equation}
while the number of solutions of Eq.~(\ref{DE3}) is given by the $n$-th complete Bell polynomial:
\begin{equation}
\nu _{r}(n)=\frac{1}{n!}B_{n}(1!d_{1},\ldots ,n!d_{n}).
\label{c3}
\end{equation}
\end{theorem}
\begin{proof}
The expansion of generating function in the neighborhood of $z=0$ 
determines the expansion of its logarithm: 
\begin{equation}
1+\sum_{k=1}^{\infty } c_{lk} z^{k} = \exp(\sum_{k=1}^{\infty } d_{lk} z^{k}).
\label{c1}
\end{equation}
We write the contour integral of both sides of this equation. The Ward identity leads to 
the relationship
\begin{equation}
c_{ln}=d_{ln}+\frac{1}{n}\sum_{k=1}^{n-1}k d_{lk}c_{ln-k},
\label{c2}
\end{equation}
which is commonly used in probability theory for calculation of cumulants. This recursion is bilateral:
it allows to find expansion coefficients of the left side of Eq.~(\ref{c1})
in terms of expansion coefficients of the right side, and \textit{vice versa}
(allowing thereby to compute $\nu _{r}(n)$ with at most $O(rn^2)$ operations).
Expanding the exponential representation of $\varphi(z)$ 
over the complete Bell polynomials, we obtain the expression (\ref{c3}).
\end{proof}

\begin{corollary}
Substituting the logarithmic derivative of (\ref{char}) into Eq.~(7), we obtain
\begin{equation}
\nu_{r}(n) = \frac{1}{n}\sum_{k=1}^{n}k d_{k}\nu_{r}(n - k).
\label{c5}
\end{equation}
\end{corollary}

\begin{remark}
Equation (\ref{c3}) solves the recursions (\ref{RE1}), (\ref{RE2}), (\ref{RE3}), and (\ref{c5}).
These recirsions appear as identities for the complete Bell polynomials.
\end{remark}

\begin{acknowledgement}
The author is grateful to Dr. R. Mahmoudvand and Dr. A. Matvejevs for communication on Theorem 1.
This paper is devoted to the 70-th anniversary of Prof. Apolodor Raduta.
\end{acknowledgement}

\bibliographystyle{amsplain}

\end{document}